\documentclass[11pt]{amsart}
\usepackage{amssymb}
\usepackage{amsmath}
\usepackage{amsmath}
\usepackage{graphicx}
\usepackage{amsfonts}
\usepackage{amssymb}
\usepackage{stmaryrd}
\usepackage{amscd}
\numberwithin{equation}{section}

\def\s{\sigma}

\def\a{\alpha}

\newtheorem{thm}{Theorem}[section]
\newtheorem{lem}[thm]{Lemma}

\begin{document}

\title[generalized ANNNI model]
{Exact solution of a generalized ANNNI model on a Cayley tree}

\author{U. A. Rozikov}
\address{Utkir Rozikov\\
Institute of Mathematics and Information Technologies\\
29, Do'rmon Yo'li str.\\
100125, Tashkent, Uzbekistan\\ email: {\tt rozikovu@yandex.ru}}

\author{H. Akin}
\address{Hasan Akin\\
Zirve University\\
Faculty of Education\\
Department of Mathematics\\
Kizilhisar Campus\\
27260, Gaziantep, Turkey\\
email: {\tt hasanakin69@gmail.com}}

\author{S. U\~guz}
\address{Selman U\~guz\\
Department of Mathematics\\
Arts and Science Faculty\\
Harran University\\
Sanliurfa, 63120, Turkey\\
email: {\tt selmanuguz@gmail.com}}

\maketitle

\begin{abstract}
We consider the Ising model on a Cayley
tree of order two with nearest neighbor interactions and competing next nearest neighbor interactions restricted to spins belonging to the same branch of the tree. This model was studied by Vannimenus and found a new modulated phase, in addition to the paramagnetic, ferromagnetic, antiferromagnetic phases and a (+ + - -) periodic phase.  Vannimenus's results based on the recurrence equations (relating the partition function of an $n-$ generation tree to the partition function of its subsystems containing $(n-1)$ generations) and most results are obtained numerically.
In this paper we analytically study the recurrence equations and obtain some exact results: critical temperatures and curves,
number of several phases, partition function.
\newline


\end{abstract}

\section{Introduction}

The model considered by Vannimenus \cite{V} consists of Ising spins $(\s=\pm 1)$ on a Cayley tree of branching ratio 2, so that every spin has three nearest-neighbor (NN). Two kinds of bonds are present: NN interactions of strength $J_1$ and next-nearest-neighbor (NNN) interactions $J_2$, these being restricted to spins belonging to the same branch of the tree. The phase diagram described by Vannimenus contains a modulated phase, as found for similar models on periodic lattices, but the multicritical Lifshitz point is at zero temperature. The variation of the wavevector with temperature in the modulated phase is studied in detail, it is shown narrow commensurate steps between incommensurate regions. The behavior of the coherence length near the different transitions is also analyzed.

The Vannimenus's model was then generalized in many directions:

In \cite{IT} it was considered a model with the competing NN and NNN interactions Ising model on a Cayley tree but in their case it is allowed for all interbranch NNN interactions on the coordination number three which was discussed earlier in \cite{KT} and it were obtained in addition to the expected paramagnetic, ferromagnetic and antiferromagnetic phases, an intermediate range of $J_2/J_1<0$ values where the local magnetization has chaotic oscillatory glass-like behavior.

Another generalization is due to Mariz et al \cite{MTA} these authors studied the phase diagram for the Ising model on a Cayley tree with competing NN interactions $J_1$ and NNN interactions $J_2$ and $J_3$ in the presence of an external magnetic field.
 At vanishing temperature, the phase diagram is fully determined, for all values and signs of $J_2/J_1$ and $J_3/J_2$; in particular, it was verified that values of $J_3/J_2$ high enough favor the paramagnetic phase. At finite temperatures, several interesting features (evolution of reentrances, separation of the modulated region into two disconnected pieces, etc.) are exhibited for typical values of $J_2/J_1$ and $J_3/J_2$.

 The next generalization is considered in \cite{OAA}, where the lattice spin model with $Q$-component discrete spin variables restricted to having orientations orthogonal to the faces of $Q$-dimensional hypercube is considered on the Cayley tree (Bethe lattice). The partition function of the model with dipole-dipole and quadrupole-quadrupole interaction is presented in terms of double graph expansions. By analyzing the regions of stability of different types of fixed points of the system of recurrent relations (which is generalization of the Vannimenus's equations), the phase diagrams of the model are plotted. For $Q\leq2$ the phase diagram of the model is found to have three tricritical points.

  The next generalizations are considered in \cite{GAT}, \cite{GMP} and \cite{GTA}. These authors have studied the phase diagram for Potts model on a Cayley tree with competing NN interactions $J_1$, prolonged NNN interactions $J_p$ and one level
NNN interactions $J_o$. In \cite{GTA} the Potts model with
$J_o\ne 0$ is considered and it is shown that for some values of $J_o$ the multicritical Lifshitz point be at non-zero
temperature and proven that as soon as the same-level interaction $J_o$ is nonzero, the paramagnetic
phase found at high temperatures for $J_o = 0$ disappears, while Ising model does
not obtain such property.

But most results of the above mentioned works are obtained numerically. Thus it is natural to try
to get some these results by an analytical way.

In this paper we consider the same model which was considered by Vannimenus (not its generalization) and study its phases  by an analytical argument. Here we shall combine analytical arguments of papers  \cite{GR}, \cite{R}, \cite{RS}, \cite{SR}.

\section{Preliminaries}

The Cayley tree (Bethe lattice \cite{Ba}) $\Gamma^k$
of order $ k\geq 1 $ is an infinite tree, i.e., a graph without
cycles, such that from each vertex of which issues exactly $ k+1 $
edges. Let $\Gamma^k=(V, L),$ where $V$ is the set of vertices of
$ \Gamma^k$, $L$ is the set of edges of $ \Gamma^k$. Two vertices
$x$ and $y$ are called {\it nearest neighbors} (NN) if there exists an
edge $l\in L$ connecting them, which is denoted by $l=\langle x,y\rangle$. A
collection of the pairs $\langle x,x_1\rangle,...,\langle x_{d-1},y\rangle$ is called a {\it
path} from $x$ to $y$. Then the distance $d(x,y), x,y\in V$, on
the Cayley tree, is the number of edges in the shortest path from
$x$ to $y$.

For a fixed $x^0\in V$ we set
\begin{equation*}
W_n=\{x\in V| d(x,x^0)=n\}, \qquad V_n=\bigcup_{k=1}^n W_k.
\end{equation*}
Denote
$$
S(x)=\{y\in W_{n+1} :  d(x,y)=1 \}, \ \ x\in W_n, $$ this set is
called a set of {\it direct successors} of $x$.

The vertices $x$ and $y$ are called {\it next-nearest-neighbor} (NNN) which is denoted by
$\rangle x, y \langle$, if there exists a vertex $z\in V$ such that $x, z$ and $y, z$ are NN. We will consider only {\it prolonged} NNN $\rangle x, y \langle$, for which there
exist $n$ such that $x\in W_n$ and $y\in W_{n+2}$, this kind of NNN is considered in \cite{V} and \cite{GMP}.

We consider Ising model with competing NN and prolonged NNN interactions on a Cayley tree
where the spin takes values in the set
$\Phi:=\{-1,1\}$, and is assigned to the vertices
of the tree \cite{V}. A configuration $\sigma$ on $V$ is then defined
as a function $x\in V\mapsto\sigma (x)\in\Phi$;
the set of all configurations is $\Phi^V$.
The (formal) Hamiltonian of the model is
\begin{equation}\label{h}
H(\sigma)=-J_1\sum_{\langle x,y\rangle\in L}
\sigma(x)\sigma(y)-J_2\sum_{\rangle x,y\langle}
\sigma(x)\sigma(y),
\end{equation}
where $J_1, J_2\in R$ are coupling constants and
$\langle x,y\rangle$ stands for NN vertices and $\rangle x,y\langle$ stands for prolonged NNN.

As usual, one can introduce the notions of Gibbs measure (phase) of the Ising model with a competing interactions
on the Cayley tree \cite{Ge}, \cite{S}, \cite{Pr}.

The standard approach consists in writing down recurrence equation relating the partition function
$$Z_n=\sum_{\s_n\in\Phi^{V_n}}\exp\{-\beta H(\s_n)\},$$
of $n$-generation tree to the partition function $Z_{n-1}$ of its subsystems containing $(n-1)$ generations.
In \cite{V} the partition function $Z_n$ of the Hamiltonian (\ref{h}) is given by
\begin{equation}\label{Z}
Z_n=\left(u_1^{(n)}+u_2^{(n)}\right)^2+\left(u_3^{(n)}+u_4^{(n)}\right)^2, \ n\geq 1.
\end{equation}
Here $u^{(n)}=\left(u_1^{(n)}, u_2^{(n)}, u_3^{(n)}, u_4^{(n)}\right)$ satisfies the following recurrent equation
\begin{equation}\label{r}
\begin{array}{llll}
u_1^{(n+1)}=a \left(bu_1^{(n)}+b^{-1}u_2^{(n)}\right)^2 \\[2mm]
u_2^{(n+1)}=a^{-1}\left(bu_3^{(n)}+b^{-1}u_4^{(n)}\right)^2 \\[2mm]
u_3^{(n+1)}=a^{-1} \left(b^{-1}u_1^{(n)}+bu_2^{(n)}\right)^2 \\[2mm]
u_4^{(n+1)}=a\left(b^{-1}u_3^{(n)}+bu_4^{(n)}\right)^2, \\
\end{array}
\end{equation}
where $a=\exp(J_1\beta)$, $b=\exp(J_2\beta)$.

Consider mapping $F:u=(u_1,u_2,u_3,u_4)\in R^4_+\to F(u)=(u_1',u_2',u_3',u_4')\in R^4_+$ defined by
\begin{equation}\label{r1}
\begin{array}{llll}
u_1'=a \left(bu_1+b^{-1}u_2\right)^2 \\[2mm]
u_2'=a^{-1}\left(bu_3+b^{-1}u_4\right)^2 \\[2mm]
u_3'=a^{-1} \left(b^{-1}u_1+bu_2\right)^2 \\[2mm]
u_4'=a\left(b^{-1}u_3+bu_4\right)^2. \\
\end{array}
\end{equation}
Then the recurrent equation (\ref{r}) can be written as $u^{(n+1)}=F(u^{(n)})$, $n\geq 0$ which in the theory of dynamical systems
is called {\it trajectory} of the initial point $u^{(0)}$ under action of the mapping $F$.
Thus  asymptotic behavior of $Z_n$ for $n\to\infty$ can be determined by
values of $\lim u^{(n)}$ i.e., trajectory of $u^{(0)}$ under action of $F$. In this paper we study the trajectory
(dynamical system) for a given initial point $u^{(0)}\in R^4$.

\section{Dynamics of $F$}
\subsection{Fixed points} In this subsection we are going to define fixed points, i.e., solutions to $F(u)=u$.

Denote Fix$(F)=\{u: F(u)=u\}$.

We introduce the new variables $\alpha=\sqrt{a}$, \ $v_i=\sqrt{u_i}$, $i=1,2,3,4$. Then the equation $F(u)=u$ becomes
as
\begin{equation}\label{v}
\begin{array}{llll}
v_1=\a \left(bv^2_1+b^{-1}v^2_2\right) \\[2mm]
v_2=\a^{-1}\left(bv^2_3+b^{-1}v^2_4\right) \\[2mm]
v_3=\a^{-1} \left(b^{-1}v^2_1+bv^2_2\right) \\[2mm]
v_4=\a\left(b^{-1}v_3^2+bv^2_4\right). \\
\end{array}
\end{equation}

\begin{lem}\label{l1}
If a vector $u$ is a fixed point of $F$ then $u\in M_1=\{u=(u_1,u_2,u_3,u_4)\in R^4_+: \ u_1=u_4,\  u_2=u_3\}$ or
$u\in M_2=\{u=(u_1,u_2,u_3,u_4)\in R^4_+: \ \sqrt{u_1}+\sqrt{u_4}=\varphi(\sqrt{u_2}+\sqrt{u_3})\}$,
where $\varphi(x)={1+a^{-1}bx\over ab+(b^2-b^{-2})x}$.
\end{lem}
\begin{proof} From (\ref{v}) we get
\begin{equation}\label{s}
\left\{\begin{array}{ll}
(v_1-v_4)[\a b(v_1+v_4)-1]+(v_2-v_3)[\a b^{-1}(v_2+v_3)]=0\\[2mm]
(v_1-v_4)[(\a b)^{-1}(v_1+v_4)]+(v_2-v_3)[1+\a^{-1}b(v_2+v_3)]=0.\\[2mm]
\end{array}\right.
\end{equation}
If $v_1=v_4$ (resp. $v_2=v_3$) from the second equation of (\ref{s}) we get $v_2=v_3$ (resp. $v_1=v_4$).
Thus $v_1=v_4$ if and only if $v_2=v_3$. Assume now
$v_1\ne v_4$ and  $v_2\ne v_3$ then system (\ref{s}) can be reduced to the equation
\begin{equation}\label{e}
(b^2-b^{-2})(v_1+v_4)(v_2+v_3)+\a b(v_1+v_4)-\a^{-1}b(v_2+v_3)-1=0.
\end{equation}
The equation (\ref{e}) gives $v_1+v_4=\varphi(v_2+v_3)$.
\end{proof}
Let us first study fixed points of $F$ which belong in $M_1$: the condition  $u_1=u_4,\  u_2=u_3$ reduces the equation $F(u)=u$ to
the following equation
\begin{equation}\label{e1}
x=f(x)\equiv a^2\left({1+b^2x\over b^2+x}\right)^2,
\end{equation}
 where $x={u_1\over u_2}$.
 Denote
 $$\tilde{a}=a^{-2}b^{-6}, \ \ \tilde{b}=b^4, \ \ y=b^2x.$$
 The following lemma gives full description of solutions to (\ref{e1}).
\begin{lem}\label{l2}(Cf. Lemma 10.7 in \cite{Pr}). Equation (\ref{e1}) has a unique positive, stable solution if
$\tilde{b}\leq 9$. If $\tilde{b} > 9$, then there exist $\nu_1(\tilde{b})$ and $\nu_2(\tilde{b})$ such that the conditions $0 < \nu_1(\tilde{b}) < \nu_2(\tilde{b})$ are satisfied and equation (\ref{e1}) has three
solutions, $x^*_1<x^*_2<x^*_3$,\, $x^*_1$ and $x^*_3$ are stable and $x_2^*$ is unstable,
if $\nu_1(\tilde{b}) < \tilde{a} < \nu_2(\tilde{b})$ and has two solutions, $x^*_1$, $x^*_2$, $x^*_1$ is unstable (saddle)
and $x_2^*$ is stable, if $\tilde{a} =\nu_1(\tilde{b})$ or
$\tilde{a} =\nu_2(\tilde{b})$. In this case, we have
$$\nu_i(\tilde{b}) ={1\over y_i}
\left({1 + y_i\over \tilde{b}+y_i}\right)^2,$$
where $y_1$ and $y_2$ are the solutions of the equation
$y^2 + (3-\tilde{b})y +\tilde{b} = 0$.
\end{lem}

Now we shall give some argument to find fixed points
of $F$ which belong in $M_2$. Again use variables $v_i$, $i=1,2,3,4$, assume $v_2+v_3=C$, with
$C>\max\{0, {\alpha b\over b^{-2}-b^2}\}$. Using Lemma \ref{l1} we get $v_3=C-v_2$ and $v_4=\varphi(C)-v_1$. Then from the first equation of
(\ref{v}) we obtain $v_2=\sqrt{b(\a^{-1}v_1-bv_1^2)}$ and from the second equation of (\ref{v}) we have
$P_4(v_1)=0$ with a polynomial $P_4$ of degree 4, coefficients of which depend on parameters $\alpha$, $b$ and $C$.
Thus a quartic equation can be obtained. Such an equations can be solved using known formulas
(see  http://mathworld.wolfram.com/QuarticEquation.html),
since we will have some complicated formulas for the coefficients and the solutions,
we do not present the solution here.

\subsection{Periodic points} A point $u$ in $R^4_+$ is called {\it periodic point} of $F$ if there exists $p$ so that
$F^p(u) =u$ where $F^p$ is the $p$th iterate of $F$. The smallest positive integer $p$ satisfying the
above is called the {\it prime period} or least period of the point $u$. Denote by Per$_p(F)$ the set of periodic points with prime period $p$.

Note that the set $M_1$ is invariant with respect to $F$ i.e., $F(M_1)\subset M_1$. In this subsection we shall
describe some periodic points of $F$ which belong in $M_1$.

Let us first describe periodic points with $p=2$ on $M_1$, in this case the equation $F(F(u))=u$
can be reduced to description of 2-periodic points of the function $f$ defined in (\ref{e1}) i.e., to solution of the equation
\begin{equation}\label{p}
f(f(x))=x.
\end{equation}

Note that the fixed points of $f$ are solutions to (\ref{p}), to find other solutions we consider the equation
$${f(f(x))-x\over f(x)-x}=0,$$
simple calculations show  that the last equation is equivalent to the following
\begin{equation}\label{x}
b^4(1+a^2b^2)^2x^2+a^2\{b^8+2(a^{-2}+a^{2})b^6+4b^4-1\}x+b^4(a^2+b^2)^2=0.
\end{equation}
The equation has two positive solutions iff $B<0$ and $D>0$ where
$$B=a^2\{b^8+2(a^{-2}+a^{2})b^6+4b^4-1\}, \ \ D=B^2-(2b^4(a^2+b^2)(1+a^2b^2))^2.$$
 We have
 $$B=\left\{\begin{array}{ll}
 >0,\ \ \mbox{if} \ \ b\geq \sqrt{\sqrt{2}-1} \ \ \mbox{or} \ \ b\leq \sqrt{\sqrt{2}-1},\, a^2\in (0,\bf{b}^-)\cup(\bf{b}^+,+\infty)\\[3mm]
 0,\ \ \mbox{if} \ \ b\leq\sqrt{\sqrt{2}-1},\ \ a^2=\bf{b}^-, \bf{b}^+\\[3mm]
 <0, \ \ \mbox{if} \ \ b\leq \sqrt{\sqrt{2}-1}, \, a^2\in (\bf{b}^-, \bf{b}^+)\\
 \end{array}\right.$$
 where
 $${\bf b}^{\pm}={1-4b^4-b^8\pm(1-b^4)\sqrt{(b^4-1)^2-4b^4}\over 4b^6};$$
 $$D=-a^2(b^4-1)^2(4b^6a^4+(3b^8+6b^4-1)a^2+4b^6)=$$
 \begin{equation}\label{D}\left\{\begin{array}{lll}
 <0, \ \ \mbox{if} \ \ \sqrt{3^{-1}}<b, \, b\ne 1 \ \ \mbox{or} \ \ b\leq \sqrt{3^{-1}}, a^2\in (0,b^-_*)\cup(b^+_*,+\infty)\\[3mm]
 0, \ \ \mbox{if} \ \ b=1 \ \ \mbox{or}\ \ b\leq \sqrt{3^{-1}},\, a^2=b^-_*,b^+_*\\[3mm]
 >0, \ \ \mbox{if} \ \ b\leq \sqrt{3^{-1}}, a^2\in (b^-_*,b^+_*)\\
  \end{array}\right.
  \end{equation}
  where $$b_*^{\pm}={1-3b^8-6b^4\pm\sqrt{(b^4-1)^3(9b^4-1)}\over 8b^6}.$$
  One can check that $\sqrt{3^{-1}}<\sqrt{\sqrt{2}-1}$, and for $b\leq \sqrt{3^{-1}}$ one has  ${\bf b}^-\leq b^-_*$ and $b_*^+\leq {\bf b}^+$.
  Thus we have proved the following
 \begin{lem}\label{l3} The solutions to (\ref{p}) which are different from fixed points of $f$ are vary as follows:

 1) If $\sqrt{3^{-1}}<b$, $b\ne 1$ or $b\leq \sqrt{3^{-1}}$, $a^2\in (0,b^-_*)\cup(b^+_*,+\infty)$ then the equation
 (\ref{x}) has no positive solution.

 2) If $b=1$ or $b\leq \sqrt{3^{-1}}$ and $a^2=b^-_*,b^+_*$ then the equation (\ref{x}) has unique
 positive solution $x_1={-B\over 2b^4(1+a^2b^2)^2}$.

 3) If  $b\leq \sqrt{3^{-1}}$, $a^2\in (b^-_*,b^+_*)$ then there are two positive solutions $x_{\pm}={-B\pm\sqrt{D}\over 2b^4(1+a^2b^2)^2}$ to (\ref{x}).
  \end{lem}
The following lemma gives useful properties of the function $f$.
 \begin{lem}\label{l4} 1) If $b>1$ then the sequence $x_n=f(x_{n-1})$, $n=1,2,\dots$ converges for any initial point $x_0>0$, where $f$ is defined in (\ref{e1}).

 2) If $b<1$ then the sequence $y_n=g(y_{n-1})$, $n=1,2,\dots$ converges for any initial point $y_0>0$, where $g(x)=f(f(x))$.
  \end{lem}
  \begin{proof} 1) For $b>1$ we have $f'(x)=2a^2(b^4-1){1+b^2x\over (b^2+x)^3}>0$ i.e.,
  $f$ is an increasing function. Here we consider the case when the function $f$ has
  three fixed points $x^*_i$, $i=1,2,3$ (see Lemma \ref{l2}. This proof is more simple for cases when $f$ has one
 or two fixed points).  We have that the point $x_2^*$
is repeller i.e., $f'(x_2^*)>1$ and the points $x_1^*, x^*_3$ are attractive i.e., $f'(x_1^*)<1$, $f'(x_3^*)<1$.
Now we shall take arbitrary $x_0>0$ and prove that $x_n =f(x_{n-1})$, $n\geq 1$
converges as $n\to\infty$. Consider the following partition $(0,+\infty) = (0, x^*_1)\cup\{x^*_1\}\cup(x^*_1, x^*_2)\cup \{x^*_2\}\cup(x^*_2, x^*_3)\cup \{x^*_3\}\cup(x^*_3,+\infty)$. For any  $x\in (0, x^*_1)$ we have $x < f(x) < x^*_1$, since $f$ is an increasing function, from the last inequalities we get  $x<f(x) < f^2(x) < f(x^*_1)=x_1^*$ iterating this argument we obtain $f^{n-1}(x)<f^{n}(x)<x^*_1$, which
for any $x_0\in (0,x^*_1)$ gives $x_{n-1}<x_n<x^*_1$ i.e., $x_n$ converges and its limit is a fixed point of $f$, since $f$
has unique fixed point $x^*_1$ in $(0, x^*_1]$ we conclude that the limit is $x^*_1$. For $x\in (x^*_1, x^*_2)$
we have $x^*_2>x >f(x)>x^*_1$, consequently $x_n > x_{n+1}$ i.e., $x_n$ converges
and its limit is again $x^*_1$. Similarly, one can show that if $x_0>x^*_2$ then $x_n\to x^*_3$ as $n\to \infty$.

2) For $b<1$ we have $f$ is decreasing and has unique fixed point $x_1$ which is repelling,
but $g$ is increasing since $g'(x)=f'(f(x))f'(x)>0$. By Lemma \ref{l3} we have that $g$ has at most
three fixed points (including $x_1$). The point $x_1$ is repelling for $g$ too, since $g'(x_1)=f'(f(x_1))f'(x_1)=(f'(x_1))^2>1$. But fixed points $x_-$, $x_+$ (see Lemma \ref{l3}) of $g$ are attractive. Hence one can repeat the same argument of the proof of the part 1) for the increasing function $g$ and complete the proof.  \end{proof}
  Lemma \ref{l3} shows that if $b>1$ i.e., $J_2>0$ then there is no any 2-periodic trajectory of $F$ on $M_1$. Moreover, the following lemma says that if $J_2>0$ then there is no any periodic trajectory on $M_1$.

  \begin{lem}\label{l5}
  1) If $J_2>0$ then for any $p\geq 2$ the equation $F^p(u)=u$ has no solution $u\in M_1\setminus {\rm Fix}(F)$.

  2) If $J_2<0$ then for any $p\geq 3$ the equation $F^p(u)=u$ has no solution $u\in M_1\setminus ({\rm Fix}(F)\cup {\rm Per}_2(F))$.
  \end{lem}
  \begin{proof} 1) Assume there is a solution $u^{(0)}\in M_1\setminus {\rm Fix}(F)$ then we get $p-$periodic trajectories
  $u^{(n+p)}_i=u^{(n)}_i$, $i=1,2$; $n=0,1,2,\dots$. Since the set $M_1$ is invariant with respect to $F$, we obtain
  $$x_{n+p}={u^{(n+p)}_1\over u^{(n+p)}_2}={u^{(n)}_1\over u^{(n)}_2}=x_n=f^n(x_0).$$ This is a contradiction, since by
  Lemma \ref{l4} we have $x_n$ is not periodic.

  2) Assume there is a solution $u^{(0)}\in M_1\setminus ({\rm Fix}(F)\cup {\rm Per}_2(F))$ then we have
  $$y_{n+p}=x_{2n+2p}={u^{(2n+2p)}_1\over u^{(2n+2p)}_2}={u^{(2n)}_1\over u^{(2n)}_2}=y_n=f^{2n}(y_0).$$ This is a contradiction, since by
  Lemma \ref{l4} we have $y_n$ is not periodic.
  \end{proof}

\section{Exact values}

Starting from random initial condition (with $u^{(0)}_1\ne u^{(0)}_4$ and $u^{(0)}_2\ne u^{(0)}_3$), one iterates the recurrence
equations (\ref{r}) and observes their behavior after large number of iterations. In the simplest situation a fixed point $u^*=(u_1^*,u_2^*,u_3^*,u_4^*)\in R^4_+$ is reached. It corresponds to (see \cite{V}):

{\it a paramagnetic phase} if $u^*\in M_1$;

{\it a ferromagnetic phase} if $u^*\in M_2$.

If the iterations give a cyclic (periodic), say with period $p$,  sequence then the corresponding
phase is called ($p$-){\it commensurate phase}. Finally, the system may remain aperiodic,
which corresponds to an {\it incommensurate phase}.

The condition $\tilde{b}>9$ of Lemma \ref{l2} requires that $J_2>0$.
Denote
$$T_c={2J_2\over \ln 3}, \ \ J_2>0.$$

Lemma \ref{l2} gives the following
\begin{thm}\label{t1} If $T\geq T_c$ then the model (\ref{h}) has unique paramagnetic phase;
If $T< T_c$ then there are exactly three (resp. two) paramagnetic phases if $(J_1, J_2)$ is such that
$b^3\sqrt{\nu_1}<a^{-1}<b^3\sqrt{\nu_2}$ (resp. $a^{-1}=b^3\sqrt{\nu_1}$ or $a^{-1}=b^3\sqrt{\nu_2}$).
\end{thm}

For the condition $b<\sqrt{3^{-1}}$ of Lemma \ref{l3} we need to condition $J_2<0$.
In this case we have $T_c={-2J_2\over \ln 3}, \ \ J_2<0.$

From Lemma \ref{l3} we get the following
\begin{thm}\label{t2} If $T\geq T_c$ then the model (\ref{h}) (on $M_1$) has unique 2-commensurate phase;
If $T< T_c$ then there are exactly two (resp. one) 2-commensurate phases if $a^2\in (b^-_*,b^+_*)$ (resp. $a^2=b^-_*$ or $a^2=b^+_*$).
\end{thm}

For a fixed temperature $T=\beta^{-1}<T_c$  we have two critical curves $a^2=b_*^{\pm}$ i.e., on terms of $J_1$ and $J_2<0$ they are
given by the following explicit relations
$$J_1= {1\over 2\beta}\ln\left(8^{-1}\{1-3e^{8J_2\beta}-6e^{4J_2\beta}\pm\sqrt{(e^{4J_2\beta}-1)^
3(9e^{4J_2\beta}-1)}\}\right)-3J_2.$$
Using Lemma \ref{l3} and formula (\ref{Z}) we can get explicit formulas for the sequence of periodic partition functions:
$$Z_n=Z_n(y)=2a^{-2/3}\times$$
$$\left\{\begin{array}{ll}
\left(\left(ab(b+{1\over by})^2+{1\over ab}({1\over b}+{b\over y})^2\right)^{-{2\over 3}}+a^{2\over 3}\left({a\over b}(by+{1\over b})^2+{b\over a}({y\over b}+b)^2\right)^{-{2\over 3}}\right)^2, \ \ n=2m\\[3mm]
 \left(\left(ab(b+{1\over bf(y)})^2+{1\over ab}({1\over b}+{b\over f(y)})^2\right)^{-{2\over 3}}+a^{2\over 3}\left({a\over b}(bf(y)+{1\over b})^2+{b\over a}({f(y)\over b}+b)^2\right)^{-{2\over 3}}\right)^2,  n=2m+1\\[3mm]
\end{array}\right.$$
where $m=0,1,2,\dots$; $y$ is one of $x_1,x_-,x_+$ defined in Lemma \ref{l3} and function $f$ is given in (\ref{e1}).

It is easy to see that if $x$ is a fixed point of $f$ then corresponding fixed point of $F$ has the form
$u^*(x)=(u^*_1(x),  u^*_2(x), u^*_2(x), u^*_1(x))$ with $u^*_1(x)=a^{-1}(b+(bx)^{-1})^{-2}$ and $u^*_2(x)=a(b+b^{-1}x)^{-2}$. If $y$ is a fixed point of $g$ then corresponding 2-periodic point of $F$ has the form
 $u^{\rm per}(y)=(u^{\rm per}_1(y),  u^{\rm per}_2(y), u^{\rm per}_2(y), u^{\rm per}_1(y))$ with
$$u^{\rm per}_1(y)=a^{-1/3}\left(ab(b+(by)^{-1})^2+(ab)^{-1}(by^{-1}+b^{-1})^2\right)^{-2/3},$$
$$u^{\rm per}_2(y)=a^{1/3}\left(a^{-1}b(b+b^{-1}y)^2+ab^{-1}(by+b^{-1})^2\right)^{-2/3}.$$
Lemma \ref{l5} gives
\begin{thm}\label{t3} The model (\ref{h}) (on $M_1$) has uncountable set $S$ of incommensurate phases $\mu_u$, where $u\in M_1\setminus ({\rm Fix}(F)\cup {\rm Per}_2(F))$. Moreover the set of incommensurate phases can be classified to (uncountable) subsets
$$S_x=\{\mu_u: \ \ u\in M_1\setminus ({\rm Fix}(F)\cup {\rm Per}_2(F))\ \ \mbox{with} \ \  \lim_{n\to\infty} F^n(u)=u^*(x)\},$$ where  $x$ is an attractive fixed point of $f$ and
$$S^{\rm per}_y=\{\mu_u: \ \ u\in M_1\setminus ({\rm Fix}(F)\cup {\rm Per}_2(F))\ \ \mbox{with} \ \  \lim_{n\to\infty} F^{2n}(u)=u^{\rm per}(y)\},$$ where  $y$ is an attractive fixed point of $g$.
\end{thm}
\section{Concluding remarks} Usually, to describe phases (Gibbs measures) of a given Hamiltonian on a Cayley tree one
has correspondence between Gibbs measures and a collection of vectors (real numbers in some particular cases) $\{h_x, \, x\in V\}$,
which satisfy a non-linear equation (see for example, \cite{GR}, \cite{GR1}, \cite{PP}--\cite{RSu}, \cite{SR}). The recurrent equation (\ref{r}) considered in this paper (which was obtained in \cite{V})
describes a vector function $\{u^{(n)}, n\in N\}$ which is a particular case of the above mentioned function $h_x$ obtained as
$h_x=u^{(n)}$ if $x\in W_n$ i.e., depends only on number of the generation set where belongs $x$ but not on $x$ itself.
Thus the solutions to the recurrent equation (\ref{r}) do not fully describe phases of the model (\ref{h}). But deriving of the functional equation
for $h_x$ corresponding to the Hamiltonian (\ref{h}) is also difficult, since there is prolonged NNN interaction. Such model can be also studied by a contour argument (see \cite{Ro} and the references therein), but this argument does not give exact solutions, in general.

By a process of iteration, for the model (\ref{h}) Vannimenus found a new modulated phase, in addition to the expected paramagnetic and ferromagnetic (fixed point) phases and a $( + + - -)$ periodic (four cycle antiferromagnetic phase, which consisted of commensurate (periodic) and incommensurate (aperiodic) regions corresponding to so called "devil's staircase". In this paper, using theory of dynamical systems  we have analytically proved many above mentioned results, i.e., the following exact results are obtained:

{\it Paramagnetic phase}: The exact critical temperature and exact critical curves are found.
It is proven that the number of the paramagnetic phases can be at most three. (Theorem \ref{t1}).

{\it Ferromagnetic phase}: We reduced description of such phases to a quartic equation (i.e., solution of the equations on $M_2$).
But we were not able to study the periodic solutions on $M_2$.

{\it Commensurate phase}: The exact critical temperature (which is obtained from the critical temperature of the paramagnetic phase by replacing $J_2$ with $-J_2$) and exact critical curves are found. On the set $M_1$ it is proven that the number of the 2-commensurate phases can be at most two and there is not $p$-commensurate phases if $p\geq 3$ (Lemma \ref{l5}, Theorem \ref{t2}). We also described exact values of periodic partition functions.

{\it Incommensurate phase}: We proved that the model has uncountably many such phases. Moreover we classified them in two classes: the first class
contains the phases which are "asymptotically fixed" (set $S_x$); the second class contains the phases which are "asymptotically periodic" (set $S^{\rm per}_y$). Note that for the usual Ising model with external field on Cayley trees such infinitely many phases are known (see \cite{Ge}, p.250).   

\section*{Acknowledgements} UAR thanks the Scientific and Technological Research Council of
Turkey (TUBITAK) for support and Zirve and Harran Universities for kind hospitality. This work was completed in the Abdus Salam International
Center for Theoretical Physics (ICTP), Trieste, Italy and UAR thanks ICTP for providing
financial support and all facilities and IMU/CDE-program for travel support.
He also thanks the TWAS  Research Grant: 09-009 RG/Maths/As-I; UNESCO FR: 3240230333.

\end{document}